\documentclass[sigconf]{aamas}  

\usepackage{flushend}

\renewcommand\footnotetextcopyrightpermission[1]{} 
\usepackage{amsthm,amsmath,amssymb}
\usepackage{amsfonts}
\usepackage{tabularx}
\usepackage{arydshln}
\usepackage{booktabs}
 \usepackage[ruled,algonl,vlined]{algorithm2e}
\usepackage{csquotes}
\usepackage{graphicx}  
\usepackage{helvet} 
\usepackage{tikz}
\usepackage[textwidth=41mm,textsize=small]{todonotes}
\usepackage{url}  
\usepackage{wrapfig}
\usepackage{xspace}

\usepackage{eqparbox}

 \newtheorem{Lemma}{Lemma}
 \newtheorem{Definition}{Definition}
 \newtheorem{Proposition}{Proposition}

\newcommand{\haris}[1]{\textcolor{red}{\textbf{Haris says:} #1}}

 \usepackage{adjustbox}  
 \newsavebox\CBox
 \def\textBF#1{\sbox\CBox{#1}\resizebox{\wd\CBox}{\ht\CBox}{\textbf{#1}}}
 \newcommand{\boldm}[1] {\mathversion{bold}#1\mathversion{normal}}
 \usepackage{tikz}

	\newcommand{\sPref}[1][]{                  
		\ifthenelse{\equal{#1}{}}{\mathrel \succ}{\mathop{P_{#1}}}
	}                                          
	\newcommand{\Indiff}[1][]{                 
		\ifthenelse{\equal{#1}{}}{\mathrel \sim}{\mathop{\sim_{#1}}}
	}
	\newcommand{\prefset}[1][]{\ifthenelse{\equal{#1}{}}{\mathcal{R}}{\mathcal{R}_{#1}}}

\usepackage{mathrsfs}
\usepackage{paralist}

\usepackage{diagbox}
\usepackage{boxedminipage}

\begin{document}
	\title{Committee Selection using Attribute Approvals}




 \author{Venkateswara Rao Kagita}
 \affiliation{%
  \institution{National Institute of Technology}
   \city{Telangana}
   \state{India}
 }
 \email{585venkat@gmail.com}
 \author{Arun K Pujari}
\affiliation{%
	\institution{Central University of Rajasthan}
	\city{Rajasthan}
	\state{India}
}
\email{akpujari@curaj.ac.in}

\author{Vineet Padmanabhan}
\affiliation{%
	\institution{University of Hyderabad}
	\city{Andhra Pradesh}
	\state{India}
}
\email{vineetcs@uohyd.ernet.in} 

\author{Vikas Kumar}
\affiliation{%
		\institution{Central University of Rajasthan}
	\city{Rajasthan}
	\state{India}
}
\email{vikas007bca@gmail.com}

\begin{abstract}
We consider the problem of committee selection from a fixed set of candidates where each candidate has multiple quantifiable attributes.
To select the best possible committee, instead of voting for a candidate, a voter is allowed to approve the preferred attributes of a given candidate.
Though attribute based preference is addressed in several contexts, committee selection problem with attribute approval of voters has not been attempted earlier.
A committee formed on attribute preferences is more likely to be a better representative of the qualities desired by the voters and is less likely to be susceptible to collusion or manipulation.
In this work, we provide a formal study of the different aspects of this problem and define properties of weak unanimity, strong unanimity, simple justified representations and compound justified representation, that are required to be satisfied by the selected committee.
We show that none of the existing vote/approval aggregation rules satisfy these new properties for attribute aggregation.
We describe a greedy approach for attribute aggregation that satisfies the first three properties, but not the fourth, i.e., compound justified representation, which we prove to be NP-complete.
Furthermore, we prove that finding a committee with justified representation and the highest approval voting score is NP-complete.
\end{abstract}

%


\keywords{participatory budgeting, proportional representation, multi-winner voting}  

\maketitle


\section{Introduction}\label{sec:intro}
Committee selection, from a set of candidates, by aggregating voters' preferences is a fundamental problem of social choice theory and has recently received considerable attention from the AI community~\cite{aziz2015,aziz2015_1,schlotter2011,skowron2015}. 
Generally, each candidate possesses a set of quantifiable attributes that make a candidate suitable or otherwise.
We consider the committee selection problem where the voter prefers a candidate because of some of the attributes possessed by the candidate, or more explicitly stated, the voter approves \emph{only} these attributes of the candidate.

Attribute-level preferences exist in various domains such as, food \cite{meuwissen2007,hadi2013}, 
health care \cite{abiiro2014}, 
housing \cite{collen2001}, farming \cite{kragt2012}, airline
services \cite{de2008}, technology product markets \cite{sriram2006}, job \cite{konrad2000}, e-
transactions \cite{butler2008}, and travel \cite{hensher1994}. 
%
%
%
%
%
Practices such as bribery~\cite{dal2007bribing,heckelman1998bribing}, and manipulation~\cite{taylor2002} \footnote{Manipulability in a voting system is a scenario wherein a voter
submits a disingenuous ballot that favors his/her interesting outcome against his true preferences~\cite{taylor2002}.} that bias in favor of one party~\cite{besley2007} can be controlled with attribute level preferences.
 
%
%
Attribute level voting system reduces these practices due to the overlap of attributes among several candidates and the non-trivial complexity of investigating these practices at attribute level. 
Therefore, the challenge is to be able to effectively aggregate the voters' approval on attributes to determine the final set of
candidates to be selected into the committee. 
Thus far, there has not been any attempt in this direction in the context of committee formation, although several other approaches for committee formation do exist.

The significance of attributes in the committee formation has been highlighted by various researchers~\cite{brams1990constrained,lang2018multi,bredereck2018} wherein these works focus on committee formation with constraint satisfaction using voters' approval ballots on candidates.
%
%
%
For instance, Brams et al. ~\cite{brams1990constrained} represent a candidate with two attributes, ``Region of the candidate'' and ``Specialty''
and the constraint could be ``10\% of candidates in the committee should be from region A''.   
%
%
Another significant work on multi-attribute
is by Lang et. al~\cite{Lang2016,lang2018multi} where they consider  proportional representation in the committee selection, but not necessarily based on voters' approvals.
%
%
Briefly, the problem in ~\cite{Lang2016,lang2018multi} is to find a committee that closely satisfies the desired proportional distribution for each of the attributes and consequently, the candidate representation takes precedence. 
%
%
Whereas, we address the problem of selecting $k$ candidates given the voters' attribute approvals apriori and therefore, the constraints in our problem setting are on the proper representation of the voters. 

%
%

Furthermore, considerable research has been done in the area of voting in combinatorial domains~\cite{lang2007vote,xia2007strongly,xia2008voting,lang2009sequential,li2011majority} to address the problem of collective decision-making over several domains or attributes given the voters/agents conditional preferences. 
%
%
For example, if voters have to agree on a common menu to be composed of main course and wine and the conditional preference of some voter could be  ``if the main course is meat then I prefer red wine, otherwise I prefer white wine''.
The works in Lang et. al~\cite{lang2007vote,lang2009sequential} address this problem by decomposing the problem into smaller problems and sequentially making decisions over individual domains. 
At each stage their approach considers voters' conditional preferences of the current domain w.r.t. the previously selected candidates of other domains to determine a score of an individual candidate. 
%
Xia et. al~\cite{xia2007strongly} defined order-independent sequential composition of voting rules and study properties of different voting rules in this context and further improvements were made by works in \cite{xia2008voting,li2011majority}. None of these schemes work for the current problem, nor can they be trivially extended. 

In this work, we address the committee selection problem wherein voters approve different attribute values. 
A candidate is represented by $d$ attribute values, which are drawn from a specified domain. 
Voters submit their approval ballots for each of $d$ attributes. 
Given the approval ballots of the voters, the objective is to select a committee $W$ of $k$ candidates.

\paragraph{Contributions}
Our key contributions are as follows.
{We present a new preference aggregation setting that simultaneously generalizes several important social choice settings. }
 First, we formally present the formulation of the problem and several analytical results of committee selection with approvals on attributes. 
Second, in the context of this problem, we revisit existing properties desirable of the final committee, such as homogeneity, consistency, monotonicity,  unanimity and justifiable representation, and provide the adapted definitions.
 We show that the two properties, namely  unanimity and justified representation are not satisfied by standard  aggregation techniques while taking into consideration attribute based approvals. 
 We focus on some key properties like \emph{weak unanimity, strong unanimity, simple justified representation} and \emph{compound justified representation}, showing some important complexity results.
Third, we propose a new aggregation rule based on greedy approach and show that this rule satisfies the unanimity and simple justified representation, but not the compound justified representation, which we prove to be NP-complete. 
 Finally, we show that many committees with justified representation are possible, but finding a justified committee with highest approval voting or highest satisfaction approval voting is an NP-complete problem. 
We propose an approximation scheme for this problem and perform worst case analysis. 

The organization of the paper is as follows. In Section 2, we formulate and analyze the committee selection problem with attribute approval voting.  In Section 3, we define standard properties of committee selection in the context of attribute approvals. The detailed analysis on unanimity and  justified representation are presented in Section 4 and Section 5, respectively. We propose a new rule called Greedy Approval Voting (GAV) in Section 6. In Section 7, we show that justified committee with highest AV, or SAV are NP-complete problems and propose an approximation scheme. 
Section 8 provides conclusions and scope for future work. 
\section{Attribute Approval Voting}
 Let $V = \{v_1, v_2,  \ldots, v_n\}$ be the set of voters and $C = \{c_1,
c_2, \ldots, c_m\}$ be the set of candidates.
Each candidate $c_i, 1\le i \le m$,
is associated with a $d$-dimensional attribute vector or simply, $d$ attributes. 
The attribute value $c_i[j]$ of candidate $c_i$ on dimension $j$ is from a domain $D^j, 1\le
j\le d$. 
Let $C_i^j$ denote the set of values in $D^j$ that are approved by voter $v_i$. 
We use $V_a$ to denote the set of voters
who have approved an attribute-value $a$. 
 The goal is to select a committee $W$ of $k$ candidates, given voters' approvals over attributes i.e., $C_i^j, \forall v_i\in V,  j\in [1,d]$. 
 Table~\ref{tab:notations} summarizes the notations used in this paper.  
 There are different aggregating rules known in the context of approval voting with candidate-approval. 
 We study below these rules in the context of attribute-approval committee selection. 
\begin{table}
\caption{Summary of notations}
\adjustbox{max width=\linewidth}{
\begin{tabular}{|l|l|}
\hline
{\bf Notation} & {\bf Description}\\ \hline
$n$ & Number of voters\\ \hline
$m$ & Number of candidates\\\hline
$d$ & Number of dimensions \\ \hline
$k$ & Target committee size\\ \hline
$V$ & Set of all voters    \\ \hline
$C$ & Set of all candidates \\ \hline
$C_i$ & Set of candidates approved by a voter $v_i$ when d=1\\\hline
$c_i[j]$ & Attribute value of a candidate $c_i$ on dimension $j$\\\hline
$D^j$ & Set of domain values on domain $j$ \\  \hline
$C_i^j$& Set of attributes approved by a voter $v_i$ on dimension $j$ \\\hline
$W^j~(resp. C^j)$ & Set of attributes of $W~(resp. C)$ on dimension $j$\\\hline
$W$ & Target committee \\ \hline
$V_i(resp.~ V'_i), i\in \mathbb{N}$ & Set of voters from $V(resp.~ V')$ who approved $c_i$ when $d=1$.\\ \hline
$V_{c_i[j]}$ & Set of voters from $V$ who approved $c_i[j]$.\\ \hline
$V_W$ & Set of voters from $V$ who approved at least one candidate in $W$\\ \hline
$c_*$ & candidate with highest utility score \\ \hline
$V_*$ & Set of voters who approved a candidate $c_*$ \\ \hline
\end{tabular}
}
\label{tab:notations}
\end{table}

{The model presented in this paper is more general than party-based approval studied recently~\citep{BGP19a}. In the latter model, each candidate can be viewed as having exactly one attribute which corresponds to the party the candidate belongs to. On the other hand, in our model, a candidate can have multiple attributes. Our model is also more general than approval-based committee voting~\citep{kilgour2010,aziz2015}. Each candidate can be viewed as having its own unique attribute. Voters can be seen as approving these attributes rather than candidates.  }

\noindent{\bf Approval Voting (AV) $-$} 
Approval Voting selects a committee $W$ that maximizes $\sum_{v_i\in V}\lvert W \cap C_i\rvert$, where $C_i$ is 
the set of candidates approved by  a voter $v_i$~\cite{brams1978}.  
When voters submit their approvals on the attributes, the \emph{ approval voting score} of $c$ with respect to $V$ is defined
 as $AV(c,V) = \sum_{v_i\in V} \frac{\sum_{j=1}^d \lvert \{v_i \mid
c[j]\in C_i^j\} \rvert}{d}$. 
The approval voting score of a committee $W$ is $AV(W,V) =
\sum_{c\in W} AV(c,V)$. 
Hence, \emph{Approval Voting (AV)} rule selects $W$ with highest $AV(W,V)$, which  can be computed by maximizing $\sum_{v_i\in V}\sum_{j=1}^d \lvert W^j\cap C_i^j\rvert$,
where $W^j$ is the set of attributes on dimension $j$ of $W$. 
%
%
Computing approval score for each attribute involves scanning $n$ ballots and can be done in $O(n_{a}n)$ where $n_{a}$ is the number of distinct attributes values over all dimensions. 
Approval score of a candidate can be computed in $O(md)$ and identifying top-$k$ candidates takes $O(klog(m))$. 
Hence, the complexity of AV for attribute-approval is $O(n_{a}n)$, which is polynomial time complexity.

\vspace{0.05cm}
\noindent{\bf Satisfaction Approval Voting (SAV) $-$}
SAV~\cite{brams2014} selects a committee $W$ that maximizes voters 
satisfaction score $ \sum_{v_i\in V}\frac{\lvert W \cap C_i\rvert}{\lvert C_i \rvert}$. 
In the case of attribute-approval voting, we define satisfaction score of  $v_i$ for $W$ as $SAV(W, {v_i}) =
\frac{AV(W, {v_i})}{min(AV(C,{v_i}), kd)} = \frac{\sum_{j=1}^d \lvert
W^j\cap C_i^j\rvert}{min(\sum_{j=1}^d \lvert
C^j\cap C_i^j\rvert, kd)}$ where $W^j~(resp. C^j)$ is the set of attributes of $W~(C, respectively)$ on
dimension $j$. SAV selects the  $W$ that maximizes $\sum_{v_i\in V} SAV(W, {v_i})$.
The complexity of  computing SAV is the same as that of computing AV.

\vspace{0.05cm}
\noindent{\bf Reweighted Approval Voting (RAV) $-$ } 
At every stage, RAV reweighs voter's approval score
of a candidate and selects the candidate with highest approval score~\cite{brams2014}.
We define
reweighed score $RAV(c, v_i)$ as $ r(v_i) \times \frac{\sum_{j=1}^d \lvert\{ v_i
\mid c[j]\in C_i^j \}\rvert}{d}$ where 
$r(v_i) =  \frac{1}{1 + AV(W, {v_i})}$. 
We start with $W=\emptyset$ and at every stage we select a candidate $c$ that
maximizes $\sum_{v_i \in V} RAV(c, v_i)$ till $\lvert W\rvert = k$. 
RAV is a multi-stage AV, and hence, the score computation needs to be done $k$ times. 
Therefore, the overall computation required in RAV is $O(n_{a}nk)$.

\vspace{0.05cm}
\noindent{\bf Proportional Approval Voting (PAV) $-$}
The objective of PAV~\cite{kilgour2010}  is to maximize the sum of voters' utilities, where utility of voter $v_i$ is $1 + \frac{1}{2} + \ldots + \frac{1}{\lvert W \cap C_i\rvert}$. With attribute-approval, 
PAV selects $W\subseteq C$ of size
$k$ that maximizes $\sum_{v_i \in V} u(AV(W, v_i))$ where $u(p) = 1+\frac{1}{2}+\ldots + \frac{1}{\lfloor p\rfloor}+\frac{1}{\lceil p\rceil}(\lceil p\rceil - p)$. 
PAV is known to be NP-hard~\cite{aziz2015_1,skowron2016}.

\vspace{0.05cm}
\noindent{\bf Minimax Approval Voting (MAV)$-$} MAV~\cite{brams2007minimax}
selects a committee $W$ that minimizes the maximum Hamming distance between 
$W$ and voters' approval ballots. 
%
We  define MAV-score of a committee $W$ 
as $Max (f(W, v_1), f(W, v_2), \ldots, f(W, v_n))$ where 
$f(W, v_i) = Max( {\bf d}(W^j, C_i^j)_{j=1}^d)$ 
and ${\bf d}(A^j,B^j) = \frac{ \lvert A^j\setminus B^j \rvert
+ \lvert B^j\setminus A^j\rvert}{\lvert D^j\rvert}$. 
MAV returns a committee $W$ with the lowest MAV-score.
MAV is also known to be NP-hard problem~\cite{legrand2007}.
\section{Properties}


%
In this section, we review some standard properties that are desired to be satisfied by multi-winner approval based rules~\cite{elkind2017properties}. 
Table~\ref{tab:notations1} summarizes the different properties satisfied by different rules. 

\begin{table}
\caption{Summary of the properties satisfied by rules}
\adjustbox{max width=\linewidth}{
\begin{tabular}{|l|c|c|c|c|c|}
\hline
 & \textBF{AV} & \textBF{SAV} & \textBF{PAV} & \textBF{RAV} & \textBF{MAV} \\ \hline
\textBF{Homogeneity} & $\checkmark$ & $\checkmark$ & $\checkmark$ & $\checkmark$ & $\checkmark$    \\ \hline
\textBF{Consistency} & $\checkmark$ & $\checkmark$ & $\checkmark$ & $\checkmark$ & $\checkmark$    \\ \hline
\textBF{Monotonicity} & $\checkmark$ & $\checkmark$ & $\checkmark$ & $\checkmark$ & $\checkmark$    \\ \hline
\textBF{Committee Monotonicity} & $\checkmark$ & $\checkmark$ & {\boldm $\times$} & $\checkmark$ & {\boldm $\times$}    \\ \hline
\end{tabular}
}
\label{tab:notations1}
\end{table}

\noindent\emph{\textbf{Homogeneity $-$}} 
A rule is said to satisfy Homogeneity property if it selects the same $W$ independent of number of times  voters' ballot $\mathcal{B} = \{C_i^j, \forall v_i \in V, j\in[1,d]\}$ is replicated. 

\noindent\emph{\textbf{Consistency $-$}} 
A rule is consistent if it satisfies the following implication. 
If the winning committee is the same $W$ w.r.t voters lists $V$ and $V'$ individually then it should be the same $W$ with respect to the voters list $V\cup V'$.
%

\noindent\emph{\textbf{Monotonicity}$-$} 
A rule is monotonic if it satisfies the following two conditions, 1) If $c\in W$ with respect to $V$ then $c\in W$ with respect to $V_{c[j]} \leftarrow V_{c[j]}\cup V_x, V_x\nsubseteq V_{c[j]}, j\in [1,d] $ 2) If $c\notin W$ with respect to $V$ then $c\notin W$ with respect to $V_{c[j]} \leftarrow V_{c[j]}\setminus V_x, V_x\subseteq V_{c[j]}, j\in [1,d]$.

\noindent\emph{\textbf{Committee Monotonicity $-$}} Suppose $W$ and $W'$ are the committees selected by rule $R$ with 
$\lvert W \rvert  = k$ and $\lvert W' \rvert  = k+1$. 
The rule $R$ is committee monotonic if $W\subset W'$.

Besides {the} above properties, \emph{Unanimity} and \emph{Justified Representation} 
are very important properties that are desired to be satisfied by approval voting based rules. 
We study these 
two properties in subsequent sections. 
\section{Unanimity}\label{sec:unanimity}
In vote based committee selection, unanimity refers to an {agreement} by all voters.
Specifically, satisfying the property means that, if there {exists} a set of candidates who are unanimously approved by all voters then at least one of them should {be} present in the selected committee $W$. 
A rule is unanimous if it selects a committee $W$ such that 
$\underset{{v_i\in V}} {\cap} C_i \cap W \ne \emptyset$  when  $\underset{{v_i\in V}} {\cap} C_i \ne \emptyset$~\cite{aziz2015}. 

Using the above definition,  we give two definitions of unanimity for attribute level committee selection as follows: 
1) \emph{Weak Unanimity} $-$  If $\exists j\in [1,d]$ with $\underset{{v_i\in V}}{\cap} C_i^j \ne \emptyset$
then $\exists j'\in [1,d]$ with $\underset{{v_i\in V}}{\cap} C_i^{j'} \cap W^{j'} \ne \emptyset$ where $W^j$ is set of attributes of $W$ on dimension $j$. 
2)\emph{ Strong Unanimity $-$ } $\forall j\in [1,d]$ with $\underset{{v_i\in V}}{\cap} C_i^j \ne \emptyset$ it holds $\underset{{v_i\in V}}{\cap} C_i^j \cap W^j \ne \emptyset$.
We note that for $d=1$ both weak and strong unanimity\footnote{Going further, if there {exists} multiple unanimous candidates, then all of them should be present in the committee $W$.} convey the same meaning. 
%
 

\begin{Lemma}
 \label{lem:unanimity}
For $k<d$, there may not exist a committee that satisfies strong unanimity.
\end{Lemma}
\noindent\emph{Proof.}
 Consider two candidates $c_1 = [a_1, b_1]$ and $c_2 = [a_2, b_2]$, and $k=1$. The approvals 
 for each of these attributes are given as $V_{a1} = V_{b2} = V$ and $V_{a2} = V_{b1} = \{v_1\}$. 
 Selecting any one of these 
 two candidates violates the strong unanimity property. 
 But, one can assure a committee $W$ that provides strong unanimity when $k\ge d$.
\hfill\qedsymbol{}

%


\begin{Lemma}
\label{lem:unanimity2}
AV, SAV, RAV, PAV, and MAV do not satisfy weak unanimity for $k\ge 1$ and $d>1$ and they satisfy weak unanimity for $d=1$.
\end{Lemma}
\noindent\emph{Proof.}
Let $X^1 = V\setminus\{v_n\}$ and $X^2 = V\setminus \{v_1\}$. $V_{c_1[1]} = V, (V_{c_1[j]} = v_1)_{j=2}^d, ((V_{c_i[j]} = X^1)_{i=2}^{\lceil m/2\rceil})_{j=1}^d $ and $((V_{c_i[j]} = X^2)_{i=\lceil m/2\rceil+1}^{m})_{j=1}^d$. 
AV, or SAV, or RAV, or PAV, or MAV selects a set $W\subseteq C\setminus\{c_1\}$, whereas $c_1[1]$ is the only attribute, which is unanimously approved by all voters and is not part of $W$. 
Hence, AV, SAV,  RAV,  PAV, and MAV do not satisfy weak unanimity for $d>1$.  {When $d=1$, the analysis is same as given in~\cite{aziz2015}}.
%
%
%
%
%
%
\hfill\qedsymbol{}
%

\begin{Proposition}
\label{prop:unanimity}
  If a rule does not {satisfy} weak unanimity then it does not satisfy strong unanimity {as well}. Hence, none of the extended rules satisfy unanimity.
\end{Proposition}
%
 %
 %
%
\section{Justified Representation}\label{sec:jr}
Justified representation\footnote{Similar notions are, Extended Justified Representation~\cite{aziz2015}, Proportional Justified Representation~\cite{sanchez2017}, Proportional Representation~\cite{pereira2016} and Strong Proportional Representation~\cite{pereira2016}.} is a crucial property that is desired to be satisfied by approval based rules. 
For satisfying this representation, if there exists a sizeable group of voters with  common preferences then the group should have {representation} in the committee. 
%
%
Based on this, we provide two definitions of justified representation for attribute-approval voting as follows: 1) Simple Justified Representation (SJR) and 2) Compound Justified Representation (CJR). 
For $d=1$ simple and compound justified representations are the same. 

\begin{Definition}
A committee $W$ satisfies simple justified representation if $\forall V'\subseteq V$: 
$
(\lvert V' \rvert \ge \frac{n}{k})\wedge(\exists j\in [1,d]~ with \underset{v_i\in V'}{\cap}~ C_i^{j}  \ne \emptyset) \implies (\exists j'\in [1,d]~ with~ \underset{v_i\in V'}{\cup}C_i^{j'}\cap W^{j'}\ne \emptyset).
$

\end{Definition}

For $k=1$, any random candidate also satisfies simple justified representation unless there exists a candidate who has no approval on any attribute. 
\begin{Lemma}
\label{lem:SJR1}
  Approval voting does not {satisfy} SJR for $k\ge2$ and $d\ge 2$
  or when $k\ge3$ and $d\ge1$.
\end{Lemma}
\noindent\emph{Proof.}
 Let $C = \{c_1, c_2, \ldots, c_{k+2}\}$ be the set of candidates, $ X^1 = \{v_i, i \in [1,n-\frac{n}{k}]\}$, and
 $X^2 =\{v_i, i\in [n-\frac{n}{k}+1, n]\}$. 
 %
 Consider profiles $((V_{c_i[j]} = X^1)_{i=1}^k)_{j=1}^d$,
 $(V_{c_{k+1}[j]} = \{X^2\})_{j=1}^{\frac{d}{2}}$, $(V_{c_{k+1}[j]} = v_n)_{\frac{d}{2}+1}^d$, 
 $(V_{c_{k+2}[j]} =  v_n)_{j=1}^{\frac{d}{2}}$ and $(V_{c_{k+2}[j]} = \{X^2\})_{j=  \frac{d}{2}+1}^d$. 
 Candidate $c_{k+1}$ or $c_{k+2}$ has to be present in $W$ in order to satisfy SJR whereas AV selects $\{c_1, \ldots, c_k\}$. 
 It is proven in~\cite{aziz2015} that AV fails JR for $k\ge 3$, and when $d=1$, SJR is same as JR. 
 Hence, AV fails SJR for $k\ge3$ and $d\ge 1$.
\hfill\qedsymbol{}
%
%
\begin{Lemma}
\label{lem:SJR2}
  PAV and RAV do not satisfy SJR for $k\ge3$ {and $d\ge 2$} or $d\ge 3$ {and  $k\ge2$}. For $k\le2$ and $d\le 2$, PAV and RAV satisfy SJR. 
\end{Lemma}
\begin{proof}
 We omit the generalized proof due to its complexity. Let $C = \{c_1, c_2, \ldots, c_5\}$,  $n = 90$, $d=2$, and $k=3$, attribute-wise voter lists are 
 $(V_{c_1[j]}= \{v_i\}_{i=1}^{35})_{j=1}^d, (V_{c_2[j]}= \{v_i\}_{i=21}^{55})_{j=1}^d, (V_{c_3[j]}  = \{v_i\}_{i=26}^{50})_{j=1}^d, (V_{c_4[j]} =\{v_1\})_{j=1}^{\frac{d}{2}}$, $(V_{c_4[j]}=  \{v_i\}_{i=61}^{90})_{j=\frac{d}{2}+1}^{d}, (V_{c_5[j]}= \{v_i\}_{i=61}^{90})_{j=1}^{\frac{d}{2}},(V_{c_5[j]}= \{v_1\})_{j=\frac{d}{2}+1}^{d} $.
 PAV (or, RAV) selects $\{c_1, c_2, c_3\}$ and ignores  set of $\frac{n}{k}$ voters who jointly approved $c_4[2]$ and $c_5[1]$.
 To extend the proof to $k>3$, we take $k-3$ additional candidates and  $(k-3)\times 30$ additional voters,  and assign 30 unique votes to each attribute of a new candidate.
%
 For $k=2$ and $d\ge3$: Let  $(V_{c_1[j]}= V_{c_2[j]}= \{v_i\}_{i=1}^{30})_{j=1}^d, V_{c_3[1]} = \{v_i\}_{31}^{60}, (V_{c_3[j]} = \{v_{1}\})_{j=2}^d$. PAV or RAV selects $W$ as $\{c_1, c_2\}$ and disregards a set of $n/k$ voters. Hence, PAV and RAV do not satisfy SJR for $k=2$ and $d\ge3$. Similarly for $k>2$ and $d\ge3$. 
 
 For $k=2$ and $d=2$: RAV selects a candidate with highest AV score in the first iteration.  The candidate having $\frac{n}{2}$ approvals for one of its attributes will have the highest score in the second iteration. 
Similar logic works for PAV. Hence, RAV and PAV satisfy SJR for $k=2$ and $d=2$. {When $d=1$, the analysis is same as given in~\cite{aziz2015}}.
\end{proof}

\begin{Lemma}
\label{lem:SJR3}
 SAV and  MAV  do not satisfy SJR for $k\ge2$, $d\ge1$.
\end{Lemma}
\noindent\emph{Proof.}
When $d=1$, attribute-approvals can be visualized as candidate-approvals. {It} is proved for candidate-approvals that MAV and SAV do not satisfy justified representation for $k\ge2$~\cite{aziz2015}. If we replicate the voters' ballots across all the dimensions then the proof follows.  
\hfill\qedsymbol{}

\begin{Definition}
$W$ provides compound justified representation (CJR)  if $\forall~ V'\subseteq V$ and $\forall i \in [1,d]$: $(\lvert V' \rvert \ge \frac{n}{k})~ \wedge~ (\underset{v_i\in V'}{\cap}~ C_i^{j}  \ne \emptyset) \implies (\underset{v_i\in V'}{\cup}C_i^{j}\cap W^{j} \ne \emptyset).
$

\end{Definition}
\begin{Lemma}
 \label{lem:JR3}
For $k\ge 3$,  there may not exist a committee that provides compound justified representation .
 \end{Lemma}
\noindent\emph{Proof.}
 Consider a set of candidates $c_i = [a_i, b_i], i\in[1,m]$ with $m=6$. Voting approvals of attributes are
 given as $V_{a1} = V_{b4} = \{v_1, v_2\}$, $V_{a2} = V_{b5} = \{v_3, v_4\},~ V_{a3} = V_{b6} =\{v_5, v_6\}, V_{b1} = V_{b2} = V_{b3} = V_{a4} = V_{a5} = V_{a6} =  \{v_1\}$. For $k=3$, none of the three candidate committees satisfy compound justified representation. One can assure a committee that satisfies CJR for $k=2$ if we assume that each attribute is approved by at least one voter.
\hfill\qedsymbol{}

\begin{Proposition}
\label{prop:JR}
   CJR implies SJR. 
    
\end{Proposition}

\begin{Lemma}
 \label{lem:CJR1}
 Checking whether there exists a committee that provides CJR is a NP-complete problem for $k\ge 2$ 
 and $d\ge 2$. 
\end{Lemma}
\noindent\emph{Proof.}
 Given a committee, CJR satisfiability is verifiable in polynomial time, hence, the problem is in NP. 
 We reduce the set cover problem to the committee selection problem to show that the current problem is NP-hard. 
 Given a set of subsets $S_1, S_2, \ldots, S_{m'}$ with $n'$ elements, the set-cover problem is to select $k'$ subsets such that 
 they cover all the elements in $\cup_{i=1}^{m'} S_i$. 
 We take an instance of a  set cover problem with $n'$ elements and $m'$ subsets, and construct $m=m'+2kn'$ candidates with $n=kn'$ voters. 
 We take elements of a subset sum problem as voters and subsets of elements as subsets of voters. 
 Let $V^{CJR}_{1}  = \cup_{i=1}^{m'} S_i = \{v_1, v_2, \ldots, v_{n'}\}$ be the subset of voters constructed from the elements of a subset selection problem and let $V^{CJR}_{i} = \{v_{n'(i-1)+1}, \ldots, v_{n'i}\}, \forall i \in [2,k]$. 
 We construct candidates' approvals as follows using the voters list $V = \cup_{i=1}^k V^{CJR}_{i}$. 
 \\
  $Initialize: V_{c_i[j]} \leftarrow \emptyset, \forall~ i\in [1,m], j\in[1,2];$\\
  $V_{c_i[j]}\leftarrow S_i, \forall~ i\in [1,m'], j\in[1,2];$
  $V_{c_x[1]}\leftarrow \{v_i\}\cup \{V^{CJR}_{h}\setminus v_{hn'}\}, \forall~ i\in [1,n'], h\in [2,k]$, 
  where $x = m'+(h-2)n'+i$.
%
  $V_{c_x[2]}\leftarrow \{v_i\}\cup \{V^{CJR}_{h}\setminus v_{hn'}\}, \forall~ i\in [1,n'], h\in [2,k], $
  where $x = m'+(k+h-3)n'+i$.
%
  $V_{c_x[1]}\leftarrow  \{V^{CJR}_{1}\setminus v_{i}\}\cup\{v_{2n'}\}, \forall~ i\in [1,n'], $
 where $x = m'+(2k-2)n'+i$.
%
%
  $V_{c_x[2]}\leftarrow  \{V^{CJR}_{1}\setminus v_{i}\}\cup\{v_{2n'}\}, \forall~ i\in [1,n'], $
 where $x = m'+(2k-1)n'+i$.
 $
  \forall~ i\in[m+1, m+2kn'], j\in [1,2]~ and~ k\ge3,~ if~ V_{c_i[j]}=\emptyset~ then~ V_{c_i[j]} = \{v_{hn'}\}, h\in[3,k].
 $
  We can see that if there is an yes instance in the set cover problem, there will be a yes instance in our problem for $k\ge2$, otherwise not.   
  %
\hfill\qedsymbol{}

 We adopt Greedy Approval Voting (GAV)~\cite{aziz2015} and extend it to attribute level in the next section. 
 We show that  GAV satisfies unanimity, SJR, and CJR {under some assumptions}. 

\subsection{Justified Committee with highest (S)AV}

We note that, the committee with highest (S)AV may not always be a justified committee. In this section, we show that many justified committees are possible for the given voter approvals over attributes and the following two problems are NP-complete, 1) Selection of justified committee with highest AV and 2) Selection of justified committee with the highest SAV.  

The preceding discussion is applicable for both simple and compound justified representations. However, we have already proven that finding a compound justified representation committee problem is NP-complete for $k\ge 2$ and $d\ge 2$, and when d=1 compound and simple justified representations are the same. Therefore, we give the NP-completeness proof related to simple justified representation in the Lemma~\ref{lem:NP_HAV}.

\subsubsection{Justified committee with highest Approval Voting}

\emph{Approval Voting (AV)} rule satisfies \emph{Justified Representation} for $k=1 \wedge d=1$, and $k=2 \wedge d=1$ if ties are 
broken in favor of sets that provide justified representation. 
However, it is not the case for $k\ge2 \wedge d\ge2$ and $k\ge 3\wedge d=1$. %
\noindent\begin{Lemma}
 \label{lem:NP_HAV}
 Justified committee with highest AV problem is NP-complete for $k\ge 2\wedge d\ge 2$ and $k\ge 3\wedge d=1$
\end{Lemma}
\begin{proof}
Let $V_W$ be a multi-set representing all the voters who approved attributes of a committee $W$. If a voter approves multiple attributes in $W$ he would appear multiple times in $V_W$.   The following is the decision problem corresponding to the justified committee with highest AV problem. \\
\emph{Question:}
Does there exist a committee $W$ of size $k$ such that $\lvert V_W \rvert \ge \tau$ and 
$\lvert V_{c_i[j]} \setminus V_W \rvert < \frac{n}{k},~  \forall c_i[j] \in \{C\setminus W\}$ ?. 

 Given a certificate of solution, 
we can easily verify that whether or not a committee satisfies the required constraints in polynomial time. 
Hence, the problem is in NP.  
We reduce the set cover problem to the committee selection problem to show that the current problem is NP-hard. 
To reduce set cover problem, we 
take an instance of a set cover problem with $n'$ number of elements and $m'$ subsets. 
For every instance with $n'$ elements and $m'$ subsets in a subset cover problem, we have $n = 2kn'$ number of voters and $m = m' + n'$ number of candidates/subsets in a committee selection problem. 
The subsets in the committee selection problem are the set of voters who approved for each candidate. 
Let $S_i, i\in[1,m']$ are subsets from the set cover problem, $V'$ and $V''$ are disjoint voter lists with $\lvert V' \rvert = (2k-3)n'$ and $\lvert V'' \rvert = 2n'$ respectively. Let $S = \cup_{i=1}^{m'} S_i = \{v_1, v_2, \ldots, v_{n'}\}$. 
We construct the candidates' subsets/voter-lists in the following manner. We take $m'$ subsets with $n'$ elements from the set cover problem and add a set $V'$ of voters to each subset i.e., $V_{c_i[j]}\leftarrow S_i \cup V', \forall i\in [1,m'], j\in[1,d]$. 
Next, we construct $n'$ additional candidate voter-lists with one voter from $S$ in each of these lists i.e. $V_{c_i[1]} \leftarrow \{\{v_{i-m'}\} \cup \{V''_{2n'-1}\}\}~ and~ V_{c_i[j]} \leftarrow \{v_{i-m'}\}, \forall i\in [m'+1, m'+n'], j\in[2,d]$, where ${V''}_{2n'-1}$ is a proper subset of $V''$ with size $2n'-1$.  
Finally, we set $\tau$ to $d(k((2k-3)n')+n')$. It is easy to see that there exist a justified committee of size $k$ with approval count greater than or equal to $\tau$ iff there exists a set cover, for $k\ge2$ and $d\ge 2$. The same proof works for $k\ge 3$ and $d=1$ if we exclude the rules relevant for $d>1$. However, NPC proof related to $k\ge 3\wedge d=1$ is invented simultaneously and appeared in \cite{bredereck2019}.  
%
\end{proof}

\subsubsection{Justified Committee with highest SAV:}
Satisfaction Approval Voting (SAV) yields the same results as Approval Voting when all voters approve an equal number of candidates. Hence, as a direct extension, we note that finding the justified committee with the highest-SAV problem is NP-complete.

\section{Attribute Level Greedy Approval Voting}
%
%
 \emph{Greedy Approval Voting (GAV)}, shown in Algorithm~\ref{algo-GAV},  starts by setting $V'= V$ and $ W  = \emptyset$. 
 At each iteration, GAV selects a candidate $c_*$ having an attribute with the highest number of approvals with respect to $V'$ and add it to $W$.  
 GAV removes all voters who voted for at least one attribute of $c_*$ from $V'$. 
 This process is repeated till $\lvert W \rvert = k$. 
 In case the voter list $V'$ is empty when $\lvert W \rvert < k$, we set $V'$ to $V$.  
Once the voter list is empty, {random selection of candidates would satisfy the weak unanimity and SJR properties but fails many other properties}. 
%
%
%
\begin{algorithm}[h!]
\label{algo-GAV}
     \caption{Greedy Approval Voting}
     \SetAlgoLined
%
      $W \leftarrow \emptyset$; $V'\leftarrow V$\;
      \While {$\lvert W \rvert < k$} {
 	 $c_* \leftarrow \underset{c_i\in C}{Argmax}~ Max(\lvert V'_{c_i[j]}\rvert)_{j=1}^d$\;
         $W\leftarrow W \cup \{c_*\}$;   
         $C\leftarrow C \setminus \{c_*\}$;
         $V'\leftarrow V'\setminus\{\cup_j~ V_{c_*[j]}\}$;
         \lIf{$\lvert V' \rvert = 0$}  {$V'\leftarrow V$; }
      }
\end{algorithm}
\begin{Lemma}
 Greedy Approval Voting satisfies weak unanimity.
\end{Lemma}
\noindent\emph{Proof.}
 If there exists an attribute which is approved by all voters, GAV selects the corresponding candidate first. Hence, GAV satisfies weak unanimity.  
\hfill\qedsymbol{}
\begin{Lemma}
  GAV satisfies strong unanimity if ties are broken in favor {of} the candidates that provide strong unanimity.
\end{Lemma}
\noindent\emph{Proof.}
 We say that a voter is unrepresented on dimension $j$ if none of his/her approved attributes from the domain $D^j$ are present in $W$. We define the following tie-breaking rules, 1) If there {exists} multiple attributes with the same number of approvals, we select the 
 one which is having the highest number of approvals according to the unrepresentative voters of the dimension where the attribute is present. 2) If multiple candidates have a unanimous attribute, we select the one with more number of unanimous attributes.  
 Using these two rules, GAV selects a committee that has at least one unanimous attribute on every dimension (if there exists such attribute on that dimension). 
 Hence, GAV satisfies strong unanimity. 
\hfill\qedsymbol{}
\begin{Lemma}
  GAV satisfies simple justified representation. 
\end{Lemma}
 \noindent\emph{Proof.}
   A rule does not satisfy simple justified representation if it completely ignores a set of $\frac{n}{k}$ voters who jointly approved for some attribute. 
   If we prove that GAV does not leave any $\frac{n}{k}$ voters who jointly approved for {an} attribute then we can say that GAV satisfies simple justified representation.
   GAV is a multi-stage approach and at every step it selects a candidate having an attribute with maximum number of approvals with respect to unrepresented voters. 
   Even if there exist completely disjoint sets of voters, each of size  $\frac{n}{k}$, GAV can cover all such voters in $k$ steps.
   Hence, GAV satisfies simple justified representation. 
 \hfill\qedsymbol{}

%
In Section \ref{sec:jr}, we have shown that CJR is a NP-complete problem. 
However, it can be solved in polynomial time under certain assumptions.
 We consider each dimension separately and identify a set of attributes that satisfy justified representation for that dimension. 
 This can be done in polynomial time using the proposed GAV.  Let $J^i$ be the set of attributes that satisfies justified representation for dimension $i$. 
 Selection of a committee that satisfies CJR is polynomial if we assume that $J^1 \times J^2 \times \ldots \times J^d \subseteq C$.  
 In addition to unanimity and justified representation, GAV satisfies other properties that are described in Section 3. 

\section{Conclusions and Discussion}
The present work initiates a new direction of research namely, use of attribute approvals for a committee selection problem. 
We extended the existing rules for committee selection by candidate-approval voting to attribute-approval voting. 
We also analyzed the standard properties that are desirable to be satisfied by these rules.
When the rules are extended to attribute-space, most of these properties were violated. 
We proposed Greedy Approval Voting and gave a detailed analysis wherein properties like unanimity and justified representation are shown to be satisfied. 
We proved that compound justified representation and justified committee with highest AV problems are NP-complete. 
We propose an approximation algorithm to determine  a justified committee with highest AV and performed worst case analysis.

The current proposal is an initiation for a committee selection problem with attribute level preferences and there is a lot of scope for future research.
%
%
 In this work, we have seen one way of extending the voters' preferences on attributes to candidates or committees, i.e., using scoring rules. Examining different ways of transforming  voters' approvals over attributes to the committees is a good direction for future research. 
Studying properties like Proportional Representation,  Extended Justified Representation, and so forth in the context of attribute level approvals is also worth to pursue for future research.  
%
%
%
%
%
Further, exploring the Pareto-optimal set of candidates and studying different properties of those sets is a good direction to pursue. We plan to investigate these aspects in the future.

  \bibliographystyle{ACM-Reference-Format}
         
\bibliography{main.bib}

\end{document}